\documentclass[12pt]{article}
\usepackage{amsfonts,latexsym}
\usepackage{amssymb}
\usepackage{amsmath,amsthm}
\usepackage{enumerate}
\usepackage{framed}
\usepackage{verbatim}
\usepackage[algo2e]{algorithm2e}
\usepackage{algorithm}
\usepackage{algorithmic}
\usepackage{color}
\newtheorem{thm}{Theorem}[section]
\theoremstyle{definition}

\newtheorem{lem}[thm]{Lemma}

\newtheorem{tm}[thm]{Theorem}
\newtheorem{rem}[thm]{Remark}

\theoremstyle{remark}

\newtheorem*{notn}{Notation}

\usepackage{tikz}

\title{Purchasing a $C_4$ online}
\author{ Michael Anastos\footnote{e-mail: manastos@andrew.cmu.edu; Research supported in part by NSF grant CCF1555599}\\
Department of Mathematical Sciences,
\\Carnegie Mellon University,
\\Pittsburgh PA 15213.
\date{}}

\begin{document}

\maketitle

\begin{abstract}
Let $G$ be a graph with edge set $(e_1,e_2,...e_N)$. We independently associate to each edge $e_i$ of $G$ a cost ${x}_i$ that is drawn from a Uniform [0, 1] distribution.  Suppose $\mathcal{F}$ is a set of targeted structures that consists of subgraphs of $G$. We would like to buy a subset of $\mathcal{F}$ at small cost, however we do not know a priori the values of the random variables ${x}_1,...,{x}_N$. Instead, we inspect the random variables $x_i$ one at a time. As soon as we inspect the random variable associated with the cost of an edge we have to decide whether we want to buy that edge or reject it for ever.

In the present paper we consider the case where $G$ is the complete graph on $n$ vertices and $\mathcal{F}$ is the set of all $C_4$ -cycles on 4 vertices- out of which we want to buy one.
\end{abstract}
\section{Introduction}
Suppose we associate to the edges of a complete graph on $n$ vertices costs drawn independently from a Uniform [0, 1]  distribution. We then inspect the cost of the edges, one by one. Once we have inspected the cost of an edge, we must decide whether we want to pay the price of that edge and purchase it or reject it forever. If we must buy a certain number of subgraphs lying in a given set of target  structures, what is the minimum expected cost paid over all strategies? Observe that the case where the set of target structures $\mathcal{F}$ is a subset of the edges of the graph and we want to purchase a single edge from $\mathcal{F}$ is closely related to the well-studied secretary problem (see for example \cite{secr} and \cite{opt}), it is attributed to Cayley \cite{cal} by Ferguson \cite{fur} and solved by Moser \cite{mos}.

Frieze and Pegden \cite{pur} studied this problem for various target structures such as matchings, spanning trees and paths between two given vertices. They also studied the cost of purchasing a triangle and a complete graph on $r$ vertices, $K_r$. They have showed that the minimum expected cost of purchasing a triangle in the POM and ROM settings (defined later) are $\Omega\big(n^{-\frac{4}{7}}(\log n)^{-2}\big)$ and $O(n^{-\frac{4}{7}})$ respectively. From those two bounds they generalised only the bound for the minimum expected cost in the ROM setting by showing that the cost a purchasing a $K_r$ is $O(n^{-(d_r+o(1))})$ where $d_r=1/(11 \cdot 2^{r-5}-1)$. They left the analysis  of the cost of purchasing any other subgraph  as an open question.

Unquestionably, the order in which the $x_i$ are examined changes some aspects of the problem and presumably it should affect the expected price paid for a structure. In this paper we will consider the following two models.
\vspace{5mm}
\\{\bf Purchaser Ordered Online Model - POM: } In this model, at each step we are allowed to inspect any uninspected $x_i$ to see if we wish to purchase  the corresponding edge. It is therefore an on-line model where we, the purchaser, choose the order in which the costs are revealed.
\\{\bf Randomly Ordered Online Model - ROM: } In this model, the order in which we inspect the items is determined by a permutation $\pi(1),...,\pi(N)$ that is chosen uniformly at random from $S_N$. At step $t$, the random variable $x_{\pi(t)}$ is revealed and we have to decide whether we want to purchase the corresponding edge $e_{\pi(i)}$ or not.
\vspace{5mm}
\\Note that  we evaluate the purchasing price of buying a subset of $F$ over different probability spaces for the two models. We evaluate it over a product space of the values of the costs for POM but over a product space of the values of the costs and the edge-orderings for ROM. Also observe that the expected cost of purchasing a structure  in the ROM setting  is  larger than  the one in the POM model. That is because we can generate the ROM model by imposing the randomness of the edges at the POM model by ourselves. We thus seek lower bounds for the POM model and upper bounds for the ROM model.
\vspace{5mm}
\\For this paper we let $G$  be the complete graph on $n$ vertices, $N = \binom{n}{2} $
and ${x}_1,...,{x}_N$,
 be independent Uniform [0, 1] random variables that are associated with the costs of edges $e_1,...,e_{N}$ of $G$. Furthermore we let
$K_{C_4}^{POM}$ and $K_{C_4}^{ROM}$  denote the minimum expected cost over all the strategies of purchasing a $C_4$ - cycle of length 4  in the POM and in the ROM models respectively. The main result of this paper is the following.
\begin{tm}\label{main}
$$  c_1n^{-\frac{5}{9}}\log^{-\frac{4}{3}}n   \leq  K_{C_4}^{POM}\leq
K_{C_4}^{ROM}\leq c_2 n^{-\frac{5}{9}}$$
for some constants $c_1,c_2>0$.
\end{tm}
 Observe that any strategy that succeeds in  purchasing a $C_4$ will purchase at least one path of length 3 before  purchasing a $C_4$. Presumably,  the number  of distinct paths of length 3 at any optimal strategy should depend on the order of the graph $n$ and tend to  infinity as $n$ tends to infinity. In light of this assumption a large portion of this paper is devoted to proving the following result.
\begin{tm}\label{second}
Let $K_{k,P_3}^{POM}$, $K_{k,P_3}^{ROM}$ denote the optimal expected cost of purchasing $k$ distinct paths of length 3 in the POM and ROM settings respectively. Then for $n^{0.5}\leq k\leq n$,
$$ c_3 \frac{k^{0.8}}{ n\log^{0.8}n}\leq K_{k,P_3}^{POM}\leq K_{k,P_3}^{ROM}\leq c_4\frac{k^{0.8}}{n} ,$$
for some constants $c_3,c_4\geq 0$.
\end{tm}
\begin{notn}
We use $P_3$ as an abbreviation of the term \emph{path of length 3}.
\end{notn}
\section{Preliminaries}
For future reference we state the Chernoff bounds in the following form. Let $X$ be distributed as a binomial $Bin(n,p)$ random variable and let $\mu=np$. Then, for any $\epsilon>0$ we have
\begin{align}
Pr[X \leq (1-\epsilon)\mu] \leq e^{-\frac{\epsilon^2\mu}{2}}, \\
Pr[X \geq (1+\epsilon)\mu] \leq e^{-\frac{\epsilon^2\mu}{2+\epsilon}}.
\end{align}
Furthermore we are going to use the following fact. Let $X_1, X_2,...,X_r$ be independent Uniform [0, 1] random variables then, for any $\theta \geq 0$ we have
\begin{align}\label{un}
Pr(X_1+X_2+...+X_r \leq \theta) \leq \frac{\theta^r}{r!}.
\end{align}
\section{Proof of the upper bound of Theorem \ref{second}.}

\begin{lem}\label{upp}
 For $n^{0.5} \leq k \leq n$ we have,
$$K_{k,p_3}^{ROM} \leq   \frac{4{k}^{0.8}}{n}. $$
\end{lem}
\begin{proof}
For the proof of this lemma we show that the upper bound is obtained by a  strategy that buys a tree of depth 2. The strategy is the following. Consider that the first $N/3$ edges that are examined have color red, the next $N/3$ have color blue  and the last $N/3$ have color green. While inspecting  red edges we buy any edge that is adjacent to vertex 1 and has cost at most $\frac{4n_1}{n}$ until we buy $n_1$ of them.  The value of $n_1$ is going to be revealed shortly. Thereafter, we buy any blue edge that is adjacent to exactly one red edge, not incident to vertex 1 and has cost at most $\frac{4n_2}{n_1n}$ until we buy another $n_2$ edges. In the case that after we have examined the cost of an edge we are in the situation where the purchased edges span $r$ $P_3$'s and there are only $k-r$ $P_3$'s that have an unexamined edge then, we purchase the rest of the edges. Similarly, if  after examining the first 2N/3 edges we haven't purchased $k$ $P_3$'s we purchase the rest of the edges.
\vspace{5mm}
\\ The number of red edges that cost at most  $\frac{4n_1}{n}$ and the number of blue edges that cost at most $\frac{4n_2}{n_1n}$ follow $Bin(\frac{N}{3},\frac{4n_1}{n})$ and $Bin(\frac{N}{3},\frac{4n_2}{n_1n})$ distributions respectively. Hence,
Chernoff bounds(1) imply that with probability at least $1-n^{-5}$ we succeed in buying the requested number of edges in the case that $n_1,n_2 = n^{\Omega(1)}$. In addition, since  cost $x_i$ follows Uniform [0, 1] distribution
we have,  $\mathbb{E}[x_i \vert x_i \leq \chi ] = \chi/2$ for any $\chi \in[0, 1]$. Thus, in the event that we succeed in buying the requested number of edges, let it be $\mathcal{E}$, the expected cost of any red edge and the expected cost of any blue edge is $2n_1/n$ and $2n_2/(n_1n)$ respectively. Hence, the expected amount paid by the above strategy, conditioned on   $\mathcal{E}$ is
\begin{equation}\label{upper}
n_1\cdot\frac{2n_1}{n}+ n_2\cdot\frac{2n_2}{n_1n}.
\end{equation}
In the event $\mathcal{E}$, any $P_3$ that is purchased consists of a blue edge plus the red edge that is adjacent to that blue edge and incident to vertex 1 plus any other red edge. Thus, in the event $\mathcal{E}$ the above strategy purchases $n_2(n_1-1)$ $P_3$'s.
By setting $k=n_2(n_1-1)$ we choose $n_1\approx (3k^2/2)^{\frac{1}{5}}$ to minimise expression (\ref{upper}). For $n_1\approx (3k^2/2)^{\frac{1}{5}}$, the quantity given in (\ref{upper}) is bounded above by
 $\frac{3.93{k}^{0.8}}{n}$.
 Therefore,
\begin{align*}
 K_{k,p_3}^{ROM} &\leq \frac{3.93{k}^{0.8}}{n}\cdot Pr(\mathcal{E})+ N\cdot Pr(\bar{\mathcal{E}})
 \leq \frac{3.93{k}^{0.8}}{n}+ N\cdot n^{-5} \leq \frac{{k}^{0.8}}{n}.
 \qedhere
\end{align*}
\end{proof}
\section{Proof of the lower bound of Theorem \ref{second}.}
At the proof of the lower bound, given at Lemma \ref{lowhard}, we make use of Lemmas \ref{sec} and \ref{subgr}. The two Lemmas provide us w.h.p. lower bounds
of the amounts that every strategy has to pay in order to purchase the edges that span  specific subgraphs. 
\begin{notn}
For $E \subset E(G)$ we let $c(E)= \sum_{e_i\in E}x_i$ i.e. the total cost of edges in $E$. Furthermore for a subgraph $H\subseteq G$ we let $c(H)=c\big(E(H)\big)$.
\end{notn}
\begin{lem}\label{sec}
Let $A$ be the event that for every $\alpha\in [n]$ and $\beta \in [n^2]$ with $\beta \geq \alpha \log^2n (=\ell) $ there does not exist a set $F$ of $\alpha$ vertices and a set $H$ of $\beta$ edges such that every edge in $H$ is incident to a vertex in $F$ and $c(H) \leq \frac{\beta^2}{10\alpha n}.$ Then, $Pr(A)=1-o(1)$.
\end{lem}
\begin{proof}
Let $S$ be the set of all quadruples $(\alpha, \beta, F, H)\in [n] \times [n^2] \times V(G) \times E(G)$ such that
$\beta \geq \alpha \log^2n$,
$\vert F \vert=\alpha$, $\vert H \vert =\beta$ and
% for all $h \in H$ $\vert h\cap F\vert \geq 1$ .
every edge in $H$ has an endpoint in $F$.
For fixed $\alpha, \beta $ there at most $\binom{n}{ \alpha}$ sets of $\alpha$ vertices $F$ each defining at most $\binom{\alpha n }{ \beta}$ sets of $\beta $ edges $H$ such that every edge in $H$ has an endpoint in $F$.
Hence, for each pair $\alpha, \beta$ there exist at most $\binom{n}{\alpha}\binom{\alpha n}{ \beta}$ pairs $H, F$ such that $(\alpha, \beta, F, H)\in S$. In addition, for $(\alpha, \beta, F, H)\in S$, (\ref{un}) implies that
$$Pr\bigg[c(H) \leq\frac{\beta^2}{10\alpha n}\bigg] \leq \bigg(\frac{\beta^2}{10\alpha n}\bigg)^{\vert H \vert}\bigg/\vert H \vert! = \bigg(\frac{\beta^2}{10\alpha n}\bigg)^{\beta}\bigg/\beta! .$$
Therefore, by taking a union bound over the quadruples in $S$  we get,
\begin{align*}
1-Pr(A)&=Pr\bigg(\exists (\alpha, \beta, F, H)\in S: c(H)\leq \frac{\beta^2}{10\alpha n} \bigg) \leq \underset{\alpha=1}{\overset{n}{\sum}} \underset{\beta = \ell}{\overset{\alpha n}{\sum}}\binom{n}{\alpha}\binom{\alpha n}{\beta}\frac{(\frac{\beta^2}{10\alpha n})^\beta}{\beta!}
\\&\leq   \underset{\alpha=1}{\overset{n}{\sum}} \underset{\beta= \ell}{\overset{\alpha n}{\sum}} n^\alpha\bigg(\frac{e\alpha n}{\beta}\bigg)^\beta\frac{\big(\frac{\beta^2}{10\alpha n}\big)^\beta}{\big(\frac{\beta}{e}\big)^\beta}
\leq   \underset{\alpha=1}{\overset{n}{\sum}} \underset{\beta= \ell}{\overset{\alpha n}{\sum}}
\big(n^{\frac{\alpha}{\beta}}\big)^\beta
\bigg( \frac{e^2}{10}\bigg)^\beta
\\ &\leq   \underset{\alpha=1}{\overset{n}{\sum}} \underset{\beta= \ell}{\overset{\alpha n}{\sum}}\bigg(n^{\frac{1}{\log^2 n}} \frac{e^2}{10}\bigg)^{\ell}
\leq n^3 \bigg( \frac{e^{2+o(1)}}{10}\bigg)^{\log^2n}  =o(1).
\qedhere
\end{align*}
\end{proof}
\begin{lem}\label{subgr}
Let $B$ be the event that for every subgraph  $H\subseteq G$ with at least 90 edges and average degree larger than 90 we have $c(H)\geq n^{-\frac{1}{15}}$. Then, $Pr(B)=1-o(1)$.
\end{lem}
\begin{proof}
For fixed $\alpha \in[n]$, and$ \beta \in [\binom{\alpha}{2}]$ there are at most $\binom{n }{ \alpha}$ ways to choose a set of $\alpha$ vertices and thereafter, at most $\binom{\binom{\alpha}{2}}{ \beta}$ ways to choose $\beta$ edges spanned by those vertices. Hence, $G$ has at most $\binom{n}{\alpha} \binom{\binom{\alpha }{2}}{\beta} $ subgraphs consisting of $\alpha$ vertices and $\beta$ edges.
Set $g(\alpha,\beta)= n^{-\frac{4}{\beta}}\big(\frac{\alpha}{en}\big)^{\frac{\alpha}{\beta}}\big(\frac{\beta}{e\alpha}\big)^2$. Thus,  by using (\ref{un}) and taking a union bound over all the subgraphs of $G$ we have,
\begin{align}\label{ppr}
\begin{split}
Pr\big( \exists H\subseteq G  &: c(H) \leq
g(\vert V(H) \vert, \vert E(H)\vert) \big)
\leq \underset{\alpha=1}{\overset{n}{\sum}} \underset{ \beta= 1}{\overset{\alpha n}{\sum}} \binom{n}{ \alpha}\binom{\binom{\alpha}{2}}{ \beta}\frac{g^{\beta}(\alpha,\beta)}{\beta!}  \\
&\leq \underset{\alpha=1}{\overset{n}{\sum}} \underset{ \beta= 1}{\overset{\alpha n}{\sum}} \bigg(\frac{en}{\alpha}\bigg)^\alpha \bigg(\frac{e\alpha^2}{2\beta}\bigg)^\beta\bigg(\frac{e\cdot g(\alpha,\beta)}{\beta}\bigg)^\beta  \\
&\leq \underset{\alpha=1}{\overset{n}{\sum}} \underset{ \beta= 1}{\overset{\alpha n}{\sum}} \bigg(\frac{en}{\alpha}\bigg)^\alpha \bigg(\frac{e\alpha^2}{2\beta}\bigg)^\beta\bigg(\frac{e }{\beta}\cdot n^{-\frac{4}{\beta}}\bigg(\frac{\alpha}{en}\bigg)^{\frac{\alpha}{\beta}}\bigg(\frac{\beta}{e\alpha}\bigg)^2\bigg)^\beta \\
&=\underset{\alpha=1}{\overset{n}{\sum}} \underset{ \beta= 1}{\overset{\alpha n}{\sum}}2^{-\beta}n^{-4}<n^{-1}.
\end{split}
\end{align}
Let $\mathcal{H}$ be the set of all subgraphs of $G$  that have at least 90 edges and average degree larger than 90. Then, $H \in \mathcal{H}$ implies that $\vert E(H) \vert / \vert V(H) \vert >90/2=45$ (since $H$ has average degree 90) or equivalently that  $\vert V(H) \vert / \vert E(H) \vert < 1/45$. Thus, (\ref{ppr}) implies that with probability at least $1-o(1)$ for every $H \in \mathcal{H}$ we have,
\begin{align*}
c(H)
&\geq g(\vert V(H) \vert, \vert E(H)\vert)
=n^{-\frac{4}{\vert E(H) \vert}}\bigg(\frac{\vert V(H) \vert}{en}\bigg)^{\frac{\vert V(H) \vert }{\vert E(H) \vert }}\bigg(\frac{\vert E(H) \vert }{e\vert V(H) \vert }\bigg)^2 \\
& \geq n^{-\frac{4}{90}}\bigg(\frac{\vert V(H) \vert }{en}\bigg)^{\frac{1}{45}}\frac{45^2}{e^2} \geq
n^{-\frac{4}{90}}\bigg(\frac{1}{en}\bigg)^{\frac{1}{45}}\frac{45^2}{e^2}\geq n^{-\frac{1}{15}}.
\qedhere
\end{align*}
\end{proof}
\begin{lem}\label{lowhard}
Every strategy  w.h.p.\@  pays for every $k\in[n^{0.5},n]$ at least
$ \frac{k^{0.8}}{ 10^5n\log^{0.8}n}$ in order to purchase $k$ $P_3$'s in the $POM$ model. Hence, for $k\in [n^{0.5},n]$ we have,  $  K_{k,P_3}^{POM}\geq \frac{k^{0.8}}{ 10^5n\log^{0.8}n}.$
\end{lem}
\begin{proof}
Suppose we  implement a strategy $T$ in the $POM$ model and let $\mathcal{P}$ be the set of all $P_3$'s that are purchased.
Call a vertex $v$ an \emph{L-vertex} if at least $2\log^2 n$ edges incident to $v$ are purchased (L for large). Otherwise, call $v$ an \emph{S-vertex}.
 In order to lower bound the cost that $T$ may pay, we consider 4 cases based on the vertices that are incident to $P_3$'s in $\mathcal{P}$. In the cases 1, 3 and 4 we condition on the event $A\cap B$. In the event that $A\cap B$ does not occur we may assume that $T$ pays nothing. However Lemmas \ref{sec} and \ref{subgr} imply that $A\cap B$ does not  occur with probability o(1).
\vspace{3mm}
\\{\bf Case 1:} There exists $T_1\subset \mathcal{P}$ such that $\vert T_1 \vert \geq \frac{k}{4}$ and every $P_3$ in $T_1$ is adjacent to 1 or 2 L-vertices.
\vspace{3mm}
\\Let $L_1$ be the set of L-vertices that are adjacent to a $P_3$ in $T_1$ and $H$ be the set of edges that are adjacent to  a vertex in $L_1$. Set $\vert L_1\vert =r$ and $\vert H \vert  =s$. Every vertex in $L_1$ has degree at least $2\log^2n$ thus, the sum of the degrees of vertices of $L_1$ , $D(L_1)$, is at least $2r\log^2n$. On the other hand, $D(L_1)$ is at most $2s$ since every edge in $H$ contributes to the degree of at most two vertices in $L_1$. Hence,
 $s\geq r \log^2n$. Furthermore, a $P_3$ in $T_1$  is incident to 1 or 2 L-vertices, therefore
$ r \geq \frac{k}{8}$. Hence, since event $A$ occurs (see Lemma \ref{sec}) with $\alpha=r$ and $\beta= s$ we have,
$$c(H) \geq \frac{s^2}{10nr} \geq \frac{r^2\log^4 n}{10nr} \geq \frac{k\log^4 n}{80n} \geq \frac{k}{n}.$$
%Observe that if a $P_3$ is incident to more than 3 L-vertices then one of its interior vertices is an L-vertex.
\\{\bf Case 2:} There exists $T_2 \subset \mathcal{P}$ such that $\vert T_2\vert \geq \frac{k}{4}$ and any $P_3$ in $T_2$ is adjacent only to S-vertices.
\vspace{3mm}
\\Any edge whose one of its endpoints is an  S-vertex can be adjacent to or part of at most $ 2\cdot(2\log^2 n)^3+
[2\cdot(2\log^2 n)^2+ (2\log^2 n)^2]\leq 20 \log^6 n$  $P_3'$s. Hence, we can construct by picking a $P_3$ in $T_2$ and then deleting from $T_2$ all $P_3$'s that it either intersects or is adjacent to, a set $T_2'$ of at least $\frac{k}{80\log^6 n} (=h)$ vertex disjoint $P_3$'s that are incident  only to S-vertices.
\vspace{3mm}
\\Let $\mathcal{F}_{T_2}$ be the set of
all the sets of edges that span at least $h$ vertex disjoint $P_3$'s. For every
$F\in\mathcal{F}_{T_2}$ we associate a quadruple of subsets of $F$ that are as follows  (if there is more than one choice for such a quadruple we pick one of them at random). The first set, $E_1$, is a set of disjoint edges. The second set, $E_2$, consists of edges adjacent to edges in $E_1$ such that $E_1\cup E_2$ does not span a $P_3$. The third set, $E_3$, consists of edges adjacent to edges in $E_1\cup E_2$ such that $E_1\cup E_2 \cup E_3$ spans a set, $\mathcal{P}_{123}$, $\vert E_3 \vert$ vertex disjoint $P_3$'s. Moreover every $P_3$ in $\mathcal{P}_{123}$ consists of an edge from each of $E_1, E_2$ and $E_3$. The last set of edges, $E_4$, is such that $E_1\cup E_4$ spans a set, $\mathcal{P}_{14}$, of $\vert E_4\vert$ vertex disjoint $P_3$'s. Finally,  $\mathcal{P}_{123}\cup \mathcal{P}_{14}$ is also a set of vertex disjoint $P_3$'s and $\vert \mathcal{P}_{123} \vert + \vert \mathcal{P}_{14}\vert =\vert E_3\vert + \vert E_4 \vert \geq h$. An example of how we may associate edges to those 4 sets in the case $h\in[3]$ is  illustrated in Figure 1.
\begin{figure}[!h]
\centering
\begin{tikzpicture}[xscale=1]
\draw (0,0)--(0,1); \draw(0,1)--(0,2); \draw(0,1)--(1,0); \draw(1,0)--(1,1); 
\draw[fill] (0,1) circle [radius=0.03]; \draw[fill] (0,2) circle [radius=0.03];
\draw[fill] (1,0) circle [radius=0.03]; \draw[fill] (1,1) circle [radius=0.03];
\draw[fill] (0,0) circle [radius=0.03];

\draw(2.25,0.5)--(2.25,1.5); \draw(1.5,0)--(2.25,0.5)--(3,0);
\draw(2.25,1.5)--(2.25,2); \draw(1.5,2)--(2.25,1.5)--(3,2);
\draw[fill] (2.25,0.5) circle [radius=0.03]; \draw[fill] (2.25,1.5) circle [radius=0.03];
\draw[fill] (1.5,0) circle [radius=0.03]; \draw[fill] (3,0) circle [radius=0.03];
\draw[fill] (2.25,2) circle [radius=0.03]; \draw[fill] (1.5,2) circle [radius=0.03];
\draw[fill] (3,2) circle [radius=0.03];

\draw (3.5,2)--(4,1) -- (4,0)--(5.25,0)--(5.25,1)--(5.75,2);
\draw(4,2)--(4,1)--(4.5,2); \draw(4.75,2)--(5.25,1);
\draw[fill] (3.5,2) circle [radius=0.03]; \draw[fill] (4,1) circle [radius=0.03];
\draw[fill] (4,0) circle [radius=0.03]; \draw[fill] (5.25,0) circle [radius=0.03];
\draw[fill] (5.25,1) circle [radius=0.03]; \draw[fill] (5.75,2) circle [radius=0.03];
\draw[fill] (4.75,2) circle [radius=0.03]; \draw[fill] (4,2) circle [radius=0.03];
\draw[fill] (4.5,2) circle [radius=0.03];

\draw(5.25,0)--(5.75,0); \draw[fill] (5.75,0) circle [radius=0.03];

\draw[ultra thick, ->] (6,1)--(7,1);

\draw[white](7.25,0)--(7.25,1);  \draw[dotted, ultra thick](7.25,1)--(7.25,2);
\draw[dotted, ultra thick][red](7.25,1)--(8.25,0); \draw[dotted, ultra thick](8.25,0)--(8.25,1); 
\draw[fill] (7.25,1) circle [radius=0.03]; \draw[fill] (7.25,2) circle [radius=0.03];
\draw[fill] (8.25,0) circle [radius=0.03]; \draw[fill] (8.25,1) circle [radius=0.03];
\draw[fill] (7.25,0) circle [radius=0.03];

\draw[dashed, ultra thick](9.5,0.5)--(9.5,1.5); \draw[white](8.75,0)--(9.5,0.5);
\draw[dashed, ultra thick][green](9.5,0.5)--(10.25,0); \draw[blue](9.5,1.5)--(9.5,2);
\draw[dashed, ultra thick][blue](8.75,2)--(9.5,1.5); \draw[white](9.5,1.5)--(10.25,2);
\draw[fill] (9.5,0.5) circle [radius=0.03]; \draw[fill] (9.5,1.5) circle [radius=0.03];
\draw[fill] (8.75,0) circle [radius=0.03]; \draw[fill] (10.25,0) circle [radius=0.03];
\draw[fill] (9.5,2) circle [radius=0.03]; \draw[fill] (8.75,2) circle [radius=0.03];
\draw[fill] (10.25,2) circle [radius=0.03];

\draw [blue](10.75,2)--(11.25,1); \draw[dashed, ultra thick](11.25,1)-- (11.25,0);
\draw[dashed, ultra thick][green](11.25,0)--(12.5,0); \draw[white](12.5,0)--(12.5,1)--(13,2);
\draw[blue](11.25,2)--(11.25,1); \draw[dashed, ultra thick][blue](11.25,1)--(11.75,2);
\draw(12,2)--(12.5,1); \draw[fill] (10.75,2) circle [radius=0.03];
\draw[fill] (11.25,1) circle [radius=0.03]; \draw[fill] (11.25,0) circle [radius=0.03];
\draw[fill] (12.5,0) circle [radius=0.03]; \draw[fill] (12.5,1) circle [radius=0.03];
\draw[fill] (13,2) circle [radius=0.03]; \draw[fill] (12,2) circle [radius=0.03];
\draw[fill] (11.25,2) circle [radius=0.03]; \draw[fill] (11.75,2) circle [radius=0.03];

\draw(12.5,0)--(13,0); \draw[fill] (13,0) circle [radius=0.03];
\end{tikzpicture}
\caption{We may associate the edges at the left with the sets $E_1$,...,$E_4$ as illustrated at the right. The black (blue, green and red respectively) edges are associated with $E_1$ ($E_2$, $E_3$ and $E_4$ resp.). Any edge that is not associated with any of the 4 sets in not present at the right. The $P_3$'s that  are spanned by the dashed (dotted respectively) edges lie in  $\mathcal{P}_{123}$  ($\mathcal{P}_{14}$ resp.).}
\end{figure}
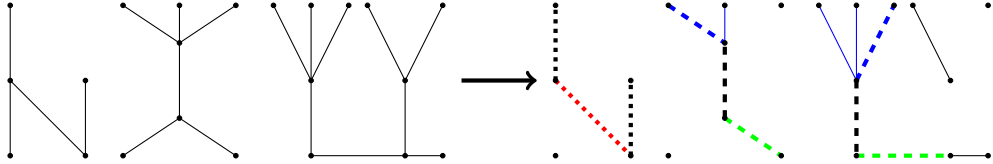  
\begin{rem}
The probabilities that follow are evaluated over the sample space $\Omega=\{$ $\text{$T$ purchases } F: F\subseteq E(G)\}$. In the case that $T$ purchases an element of $F\in\mathcal{F}_{T_2}$ then $E_1,...,E_4$ are the sets that are associated with $F$. Otherwise, we set $c(E_1)=...=c(E_4)=\infty$.
\end{rem}
 For $i\in [4]$ we  set $n_i=\vert E_i\vert$.
For fixed $n_1$, by taking union bound over all $\binom{N }{n_1}$ choices for $E_1$, (\ref{un}) implies that
\begin{align}\label{ce1}
\begin{split}
    Pr\bigg(c(E_1) &\leq \frac{n_1^2}{10n^2} \bigg) \leq
    \binom{N}{n_1}
    \frac{\big(\frac{n_1^2}{10n^2} \big)^{n_1}}{n_1!} \leq
    \bigg(\frac{eN}{n_1}\bigg)^{n_1}
\bigg(\frac{en_1}{10n^2} \bigg)^{n_1} \leq \bigg(\frac{e^2}{10} \bigg)^{n_1}.
\end{split}
\end{align}
In addition, for fixed $n_4$ and $H\subseteq E(G)$, conditioned on the event $\{E_1=H\}$, $E_4$ is a subset of the edges that are spanned by the endpoints of edges in $H(= E_1)$; hence, of at most $\binom{2n_1}{ 2}$ edges. Thus, by taking union bound over all the possibilities for $E_4$, (\ref{un}) implies that
\begin{align}\label{ce2}
\begin{split}
    Pr\bigg(c(E_4) &\leq \frac{n_4^2}{20n_1^2}\bigg|E_1=H \bigg) \leq
    \binom{\binom{2n_1}{2}}{n_4}
    \frac{\big(\frac{n_4^2}{20n_1^2} \big)^{n_4}}{n_4!} \leq
    \bigg(\frac{2en_1^2}{n_4}\bigg)^{n_4}
\bigg(\frac{en_4}{20n_1^2} \bigg)^{n_4} \leq \bigg(\frac{e^2}{10} \bigg)^{n_4}.
\end{split}
\end{align}
Recall that $n_4 \leq n_1$. Thus, by taking union bound over $n_1$ and $n_4$, (\ref{ce1}) and (\ref{ce2}) imply that
\begin{align*}
    P&r\bigg(\bigg\{c(E_1)+c(E_4) \leq \frac{n_1^2}{10n^2}+ \frac{n_4^2}{20n_1^2}\bigg\} \wedge \bigg\{n_4\geq \frac{h}{2}\bigg\} \bigg)
   % \leq  Pr\bigg(c(E_1) \leq \frac{n_1^2}{10n^2} \bigg) +
    % Pr\bigg( c(E_4) \leq \frac{n_4^2}{20n_1^2} \bigg)
    \\ &\leq \sum_{n_1=h/2}^{N}\sum_{n_2=h/2}^{N} \Bigg [
     Pr\bigg(c(E_1) \leq \frac{n_1^2}{10n^2} \bigg) + \sum_{\substack{H\subset E(G):\\|H|=n_1}}
     Pr\bigg( c(E_4) \leq \frac{n_4^2}{20n_1^2}\bigg\vert E_1=H \bigg)Pr(E_1=H)\Bigg] \\
      &\leq \sum_{n_1=h/2}^{N}\sum_{n_2=h/2}^{N} \Bigg[
    \bigg(\frac{e^2}{10} \bigg)^{n_1} + \sum_{{H\subset E(G):|H|=n_1}}
      \bigg(\frac{e^2}{10} \bigg)^{n_4}Pr(E_1=H)\Bigg]\\
     &\leq \sum_{n_1=k/(160\log^6n)}^{N}\sum_{n_2=k/(160\log^6n)}^{N} \Bigg[
    \bigg(\frac{e^2}{10} \bigg)^{n_1} +
      \bigg(\frac{e^2}{10} \bigg)^{n_4}\Bigg] \leq n^4\cdot2 \bigg(\frac{e^2}{10} \bigg)^{(1+o(1))n^{0.5}} =o(1).
\end{align*}
The fact that $E_1$ and $E_4$ are disjoint implies that $c(\mathcal{P}_{14})\geq c(E_1)+c(E_4)$.
Therefore, the calculation above implies that with probability o(1) \emph{either} case 2 does not occur \emph{or} $n_4< h/2$ \emph{or} the following inequality does not hold.
\begin{align*}
c(\mathcal{P}_{14})\geq c(E_1)+c(E_4) \geq \frac{n_1^2}{10n^2}+ \frac{n_4^2}{20n_1^2}\geq 2\frac{n_1}{\sqrt{10}n}\frac{n_4}{\sqrt{20}n_1} \geq \frac{n_4}{10n} \geq \frac{h}{20n} =\frac{k}{1600n\log^6n}.
\end{align*}
Note that if case 2 occurs then, $n_4 < h/2$ implies that $n_3\geq h/2$. Consequently, we have that  $n_1, n_2, n_3\geq n_3\geq h/2$. As before, the edges in $E_1$ are chosen from all edges in $E(G)$. Thereafter, for $H_1,H_2\subseteq E(G)$ conditioned on the event $\{E_1=F_1\}$ (and on the event $\{E_1=F_1 \wedge E_2=F_2\}$ respectively), the edges in $E_2$ ($E_3$ resp.) are chosen from all the edges that are incident to an edge in $H_1$ ($H_1\cup H_2$ resp.); hence, from at most $2n_1n$ edges  ($2n(n_1+n_2)$ resp.). Therefore, by using the same method as before and the fact that $E_1$ $E_2$ and $E_3$ are disjoint we get that with probability o(1) \emph{either} case 2 does not occur \emph{or} $n_3 < h/2$ (which implies that either $n_4\geq h/2$ or case 2 does not occur) \emph{or} the following inequality does not hold.
\begin{align}\label{cooos}
c(\mathcal{P}_{123})\geq c(E_1)+c(E_2)+c(E_3) \geq \frac{n_1^2}{10n^2}+\frac{n_2^2}{20n_1n}+\frac{h^2}{160n_2n}.
\end{align}
By setting the partial derivatives of the right hand side to zero we get $n_1^3=n_2^2n/4$ and $n_{2}^3=n_1h^2/16$. Solving for $n_1,n_2$ and substituting those values in (\ref{cooos}) we get
$$c(\mathcal{P}_{123}) \geq 0.04\bigg(\frac{h}{n}\bigg)^{\frac{8}{7}} \geq10^{-5}\bigg(\frac{k}{n\log^6 n}\bigg)^{\frac{8}{7}}.$$
Therefore, with probability o(1) case 2 occurs and $T$ pays less than $$\min\bigg\{\frac{k}{1600n\log^6n}, 10^{-5}\bigg(\frac{k}{n\log^6 n}\bigg)^{\frac{8}{7}} \bigg\}.$$
\\{\bf Case 3}: There exists $T_3\subset \mathcal{P}$ such that $\vert T_3 \vert \geq \frac{k}{4}$ and every $P_3$ in $T_3$ has exactly one of its internal vertices being an S-vertex.
\vspace{3mm}
\\Let $V_L=\{w_1,...,w_q\}$ be the set of L-vertices that are internal vertices of $P_3$'s in $T_3$. For $i\in [q]$ let $d_i$ be the degree of $w_i$. Furthermore, let $E_L$ to be the edges that are incident to a vertex in $V_L$ and are purchased by $T$.
An edge may contribute to the degree of at most two vertices in $V_L$; hence, $\vert E_L \vert\geq  (d_1+...+d_q)/2 \geq q(2\log^2n)/2=q\log^2n$.
Therefore, since event $A$ occurs (with $F=V_L, H=E_L$, $\alpha=q$ and $\beta= (d_1+...+d_q)/2$)  we get that
\begin{align}\label{cas22}
c(T_3) \geq  c(E_L) \geq \frac{\big(\frac{d_1+...+d_q}{2}\big)^2}{10qn} \geq \frac{d_1^2+...+d^2_q}{40qn}.
\end{align}
Furthermore, a vertex of degree $d$ can be an internal vertex of at most $d(d-1)(2\log^2 n )$  $P_3$'s whose other internal vertex is an S-vertex. Hence, if we let for $i\in [q]$ $c_i$  be the number of $P_3$'s in $T_3$ that have $w_1$ 
as an internal vertex we have, 
\begin{align}\label{cas21}
\frac{k}{4}\leq |T_3|\leq \sum_{i\in[q]}c_i \leq  2(d_1^2+...+d^2_q)\log^2n.
\end{align}
(\ref{cas22}),(\ref{cas21}) imply,
$$c(T_3) \geq \frac{k}{320 n \log^2 n}.$$
\\{\bf{Case 4}}: There exists $T_4\subset \mathcal{P}$ such that $\vert T_4 \vert \geq \frac{k}{4}$ and every $P_3$ in $T_4$ has both of its internal vertices being  L-vertices.
\vspace{3mm}
\\Either $T$ purchases a subgraph $H\subseteq G$ that consists of at least 90 edges and has average degree larger at most 90  or $T$ does not purchase such a subgraph. In the first case, since event $B$ occurs, we have that $c(H)\geq n^{-\frac{1}{15}}$; hence, $T$  pays at least $n^{-\frac{1}{15}}$. Thus, suppose that every subgraph $H$ of $G$ that is purchased and consists of at least 90 edges has average degree at most 90.
\vspace{3mm}
\\Let $A_0$ be the set of all  L-vertices in $G$. Furthermore let $B_1$ be the subset of L-vertices  of $A_0$ that have  at most 179 neighbours in $A_0$ and $A_1=A_0/B_1$ (i.e.\@ the subset of vertices  of $A_0$ that have  at least 180 neighbours in $A_0$).
\vspace{3mm}
\\Recursively for $i \in [\log n]$ given $A_{i-1}$   we let $B_{i}$ to be the subset of vertices  of $A_{i-1}$ that have at most 179 neighbours in $A_{i-1}$ and $A_{i}=A_{i-1}/B_{i}$.
\vspace{3mm}
\\For $i \in [\log n]$ let $H_i$ be the subgraph with vertex set $A_i\cup B_i$ and edge set all the edges that have been purchased and are spanned by $A_i\cup B_i$. Then, either $\vert E(H_i)\vert < 90$ or $\vert E(H_i) \vert \geq 90$. In the case that
$\vert E(H_i)\vert < 90$ we have that  no vertex in $V(H_i)$ has degree more than 179; therefore, $A_i = \emptyset$. On the other hand, if   $\vert E(H_i) \vert \geq  90$ then, since the average degree of the vertices in $H_i$ is at most 90 and at most half of the vertices of $H_i$ can have degree greater or equal to  two times the average degree of $H_i$ (in our case 180), we get that $\vert A_i \vert \leq 0.5 \vert V(H_i) \vert = 0.5 \vert A_{i-1} \vert.$  In both cases we get $\vert A_{i} \vert \leq \ 0.5 \vert A_{i-1} \vert $; hence, $\vert A_{\log n}\vert \leq 1$.
\vspace{3mm}
\\Observe that for each $P_3$ in $T_4$ there exists exactly one $i$ such that $B_i$ contains one of its internal vertices and the other one is found either in $A_i$ or in  $B_i$. In light of this observation, we  partition  $T_4$ into $2\log n$ sets as follows. For $i \in [\log n]$ and $C_i\in \{A_i,B_i\}$ we set $T_4(B_i,C_i):=\{(p_1,p_1p_2,p_2,p_2p_3,p_3,p_3p_4,p_4)\in T_4: p_2\in B_i \wedge p_3 \in C_i\text{ \emph{or} }  p_2\in C_i \wedge p_3 \in B_i\}.$ Since $T_4$ is the union of all $T_4(A_i,C_i)$'s and consists of at least $\frac{k}{4}$ $P_3$'s, we have that one of the $T_4(B_i,C_i)$'s must consist of at least  $\frac{k}{8\log n}$ $P_3$'s. Let $i=\min\big \{j\in [\log n]: \vert T_4(B_j,A_j)\vert \geq \frac{k}{8\log n} \text{ \emph{or} } \vert T_4(B_j,B_j)\vert \geq \frac{k}{8\log n} \big\}$
\vspace{3mm}
\\{\bf{Sub-case 1}}:  $T_4(B_i,A_i)\geq \frac{k}{8\log n}$. 
\vspace{3mm}
\\Let $A_i=\{a_{i,1},...,a_{i,n_i}\}$. For $j \in [n_i]$ we let $D_{i,j}$  be the set of edges that are purchased and are incident to $a_{i,j}$. In addition, we let $Q_{i,j}$  be the set of edges that are purchased and are incident to a neighbour of $a_{i,j}$ that lies in $B_i$ and do not lie in $D_{i,j}$ (illustrated in Figure 2).
\begin{figure}[!h]
\centering
\begin{tikzpicture}[yscale=1]
\draw[dashed] (7,1) circle [radius=0.8]; \node at (7,2.1) {$B_i$};
\draw[fill] (4,1) circle [radius=0.04]; \node at (4.1,0.75) {$a_{i,j}$};
\draw[fill] (6.9,1.6) circle [radius=0.03]; \draw[fill] (7.1,1.3) circle [radius=0.03];
\draw[fill] (6.9,1) circle [radius=0.03]; \draw[fill] (7.1,0.7) circle [radius=0.03];
\draw[fill] (6.9,0.4) circle [radius=0.03]; \draw[blue] (4,1)--(6.9,1.6);
\draw[blue] (4,1)--(7.1,1.3); \draw[blue] (4,1)--(7.1,0.7);
\draw[fill] (3,1.8) circle [radius=0.03]; \draw[fill] (3,1.5) circle [radius=0.03];
\draw[fill] (3,1.2) circle [radius=0.03]; \draw[fill] (3,0.9) circle [radius=0.03];
\draw[fill] (3,0.6) circle [radius=0.03]; \draw[fill] (3,0.3) circle [radius=0.03];
\draw[blue] (4,1)--(3,1.8); \draw[blue] (4,1)--(3,1.5);
\draw[blue] (4,1)--(3,1.2); \draw[blue] (4,1)--(3,0.9);
\draw[blue] (4,1)--(3,0.6); \draw[blue] (4,1)--(3,0.3);
\draw[fill] (8.5,1.8) circle [radius=0.03]; \draw[fill] (8.5,1.6) circle [radius=0.03];
\draw[green] (8.5,1.8)--(6.9,1.6)--(8.5,1.6);

\draw[fill] (8.5,0.4) circle [radius=0.03]; \draw[fill] (8.5,1) circle [radius=0.03];
\draw[fill] (8.5,0.8) circle [radius=0.03]; \draw[fill] (8.5,0.6) circle [radius=0.03];
\draw[green] (8.5,1)--(7.1,0.7)--(8.5,0.8); \draw[green] (8.5,0.4)--(7.1,0.7)--(8.5,0.6);

\end{tikzpicture}
\caption{ The edges in $D_{i,j}$ are in blue and the edges in $Q_{i,j}$ are in green.}
\end{figure}
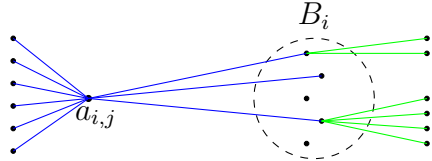
\\Observe that  $Q_{i,j}\cup D_{i,j}$
spans all $P_3$'s whose  one of their internal vertices is $a_{i,j}$ while the other one lies in $B_i$. Thus, if we let
 $q_{i,j}=\vert Q_{i,j}\vert$, $d_{i,j}=\vert D_{i,j} \vert$ and $p_{i,j}$  be the number of $P_3$'s that are spanned by $Q_{i,j} \cup D_{i,j}$, we have that $p_{i,j} \leq q_{i,j}d_{i,j}$. This is because each $P_3$ that is spanned by $Q_{i,j}\cup D_{i,j}$ is determined by its end-edges, one of which is found in $Q_{i,j}$ and the other one in $D_{i,j}$.
 \vspace{3mm}
\\In the case that $q_{i,j} \geq d_{i,j} \log^2 n$ since event $A$ occurs (see Lemma \ref{sec}), with $F$ consisting of the neighbours of $a_{i,j}$, $H=Q_{i,j}$, $\alpha= d_{i,j}$ and $\beta = q_{i,j}$ (and $F=\{a_{i,j}\}$, $H=D_{i,j}$, $\alpha=1$ and $\beta=d_{i,j}$ resp.) we get that $
c(Q_{i,j}) \geq {q_{i,j}^2}/{10d_{i,j}n}$ ( and
$c(D_{i,j}) \geq {d_{i,j}^2}/{10n}$ resp.). Hence,
\begin{align}\label{lol}
\begin{split}
2c(Q_{i,j}\cup D_{i,j}) &\geq c(Q_{i,j})+c(D_{i,j}) \geq \frac{q_{i,j}^2}{10d_{i,j}n}+\frac{d_{i,j}^2}{10n}
\geq \frac{p_{i,j}^2}{10d_{i,j}^3n}+\frac{d_{i,j}^2}{10n} \\
&\geq  \frac{p_{i,j}^2}{10\alpha^3n}+\frac{\alpha^2}{10n} \geq 0.1 \frac{p_{i,j}^{0.8}}{n},
\end{split}
\end{align}
where $\alpha=\bigg(\frac{3p^2_{i,j}}{2}\bigg)^{0.2}$ is chosen in order to minimise the expression $ \frac{p_{i,j}^2}{10d_{i,j}^3n}+\frac{d_{i,j}^2}{10n} $ over $d_{i,j}$.
\\On the other hand, in the case that $q_{i,j} \leq d_{i,j}\log^2n$ since $p_{i,j}\leq q_{i,j}d_{i,j}$ and event $A$ occurs, we have,
$$c(Q_{i,j}\cup D_{i,j}) \geq c(D_{i,j}) \geq \frac{d_{i,j}^2}{10n}\geq\frac{d_{i,j} q_{i,j}}{10n\log^2n}
\geq \frac{p_{i,j}}{10n\log^2n} .$$
In order to calculate the total cost of purchasing the set of edges spanned by the paths in $T_4(B_i,A_i)$, recall that each vertex in $B_i$ is adjacent to at most 179 vertices in $A_i$. Hence, an edge is found in at most 179 $Q_{i,j}$'s  and two $D_{i,j}$'s. Therefore, (\ref{lol}) implies
\begin{align}\label{cop}
181 c(T_4) \geq \sum_{j=1}^{n_i} c(Q_{i,j} \cup D_{i,j}) \geq \sum_{j=1}^{n_i} 0.05\frac{p_{i,j}^{0.8}}{n}  \geq0.05\frac{\big( \sum_{j=1}^{n_i} p_{i,j}\big)^{0.8}}{n},
\end{align}
where the last inequality follows from the concavity of the function $f(x)=x^{0.8}$. Combining (\ref{cop}) with the fact that  $\vert T_4(B_i,A_i) \vert = \underset{j=1}{\overset{n_i}{\sum}}p_{i,j} \geq
\frac{k}{8\log n} $ we get,
\begin{align*}
    c(T_4) \geq 10^{-4} \cdot \frac{\big( \sum_{j=1}^{n_i} p_{i,j}\big)^{0.8}}{n}
    \geq { 10^{-5}}\cdot\frac{k^{0.8}}{n\log^{0.8}n}.
\end{align*}
\\{\bf{Sub-case 2}}: $T(B_i,B_i)\geq \frac{k}{8\log n}$ .
\vspace{3mm}
\\We let $B_i'$ be the subset of vertices of $B_i$ that have a neighbour in $B_i$. Thereafter, sub-case 2 follows similarly to sub-case 1 (replace both sets $A_i$ and $B_i$ with $B_i'$ ).
\vspace{5mm}
\\Summarising the cost of purchasing $k\geq n^{0.5}$ $P_3$'s  paid by any strategy is w.h.p.\@ at least
\begin{align*}
& \min\bigg\{
\frac{k}{n},
 \frac{k}{1600n\log^6 n},
 10^{-5}\bigg(\frac{k}{n\log^6n}\bigg)^{\frac{8}{7}}
\frac{k}{320  n \log^2 n},n^{-\frac{1}{15}},
\frac{k^{0.8}}{10^5n\log^{0.8}n} \bigg\}
=\frac{k^{0.8}}{10^5n\log^{0.8}n}.
\qedhere
\end{align*}
\end{proof}
\section{Purchasing a $C_4$}
\emph{To upper bound $K_{C_4}^{ROM}$} we follow the strategy described in Lemma \ref{upp} in order to purchase $k$  $P_3$'s using only the first $\frac{2N}{3}$ edges  and paying in expectation less than $4\frac{k^{0.8}}{n}$. Note that the $k$ $P_3$'s that we purchase have distinct pairs of endpoints.
Let $END$ be the set of edges that join those pairs of endpoints and are included in  the last $\frac{N}{3}$ edges. Then, we want to buy a single edge from $END$. In order to bound the minimum  expected cost of purchasing an edge from $END$ we proceed by bounding from below the cardinality of $END$.
\vspace{3mm}
\\With $k=n_1n_2$ there at most $t=\binom{n }{ n_1}\binom{ n }{ n_2}$ trees with root the vertex 1, $n_1$ vertices at depth 1 and $n_2$ vertices at depth 2. For each such tree there are at
most $d=\binom{k}{0.9k}\binom{N-0.9k}{\frac{2}{3}N-0.9k}=\binom{k}{0.1k}\cdot\binom{N-0.9k}{\frac{2}{3}N-0.9k}$
ways to choose $2N/3$ out of $N$ edges  such that at most $0.1k$ edges of the edges in $END$ are not chosen. Hence for $k=n^{\Omega(1)}$, with $n_1\approx (3k^2/2)^{0.2}$ (as found in Lemma \ref{upp}),  the probability that $END \leq  0.1k (=n_1n_2)$ is bounded by
\begin{align*}
{t\cdot d}\bigg/\binom{N}{ \frac{2}{3}N} 
&\leq n^{n_1}n^{n_2}\bigg(\frac{ek}{0.1k}\bigg)^{0.1k}\frac{(N-0.9k)!}{(\frac{2}{3}N-0.9k)!(\frac{1}{3}N)!} \bigg/ \frac{N!}{(\frac{2}{3}N)!(\frac{1}{3}N)!} \\
& \leq n^{n_1+n_2}(10e)^{0.1k}\prod_{i=0}^{0.9k-1}\frac{2N/3-i}{N-i}
\\& \leq   n^{n_1+n_2}(10e)^{0.1k} 0.67^{0.9k} \leq n^{2k^{0.6}}0.98^k=o(1).
\end{align*}
To purchase an edge from $END$ we implement the following strategy. While examining the $\ell$th edge from $END$ we purchase it if it costs at most $\frac{2}{END-\ell}$ and we have not purchased any other edge from $END$.   Denote the expected cost of the strategy above when $\vert END \vert=m$ by $c(m)$. We will show by induction that $c(m)\leq \frac{2}{m}$.
\vspace{3mm}
\\For $m=1$ it is trivial. Assume it is true for  $m=s-1$. Therefore, if $|END|=s$ then, given that we do not purchase the first edge that we examine, the expected amount that we pay equals to the one that we pay when $s-1$ edges are examined, which is at most $c(s-1)$. For $m=s$ let $x_{s}$ to be the first cost that is examined. Set $r_s=\frac{2}{s-1}$ then,
\begin{align*}
    c(s)&= \mathbb{E}\big(x_s \vert x_s\leq r_s\big)P\Big(x_s\leq r_s\big) +
 \mathbb{E}\big(c(s-1)\vert x_s> r_s\big)P\Big(x_s> r_s \big)\\
 &= \frac{1}{s-1}\cdot\frac{2}{s-1}+\frac{2}{s-1}\bigg(1-\frac{2}{s-1}\bigg)=\frac{2}{s-1}\cdot \frac{s-2}{s-1} \leq \frac{2}{s}.
 \end{align*}
 Therefore, our strategy  pays at most $4\frac{k^{0.8}}{n}$ in order to buy a $k \geq n^{0.5}$ $P_3$'s plus at most $\frac{2}{0.1k}$ in expectation in order to purchase the fourth edge of a $C_4$. Hence,
 $$K_4^{ROM} \leq \underset{ k \in[n^{0.5}, n]}{\min}\bigg\{4\frac{k^{0.8}}{n}+\frac{20}{k} \bigg\}\leq 17n^{-\frac{5}{9}}.$$
\emph{To lower bound  $K_{C_4}^{POM}$ } we suppose that we implement a strategy $T$. While executing  $T$ we keep two lists (by only adding elements), one list of $P_3$'s, call it $L_P$ and one list of edges, call it $L_E$. At some stage of the algorithm, suppose that $L_P=\{P_1,...,P_k\}$ and $L_E=\{f_1,...,f_k\}$, then, the following are satisfied. First, $L_P$ consists of all $P_3$'s that are spanned by the purchased edges. Second, if $i<j\leq k$ then the edges that span $P_j$ have not been purchased before the edges that span $P_i$. Finally, the edge  $f_i$ is spanned by the endpoints of $P_i$.
\vspace{3mm}
\\ \emph{Claim}: W.h.p\@ there does not exist $\ell \in [n]$ such that $c(f_\ell) \leq \frac{1}{\ell \log^2 n}$.
\vspace{3mm}
\\\emph{Proof of the claim}: Let $f_1,...,f_n$  be a sequence of edges. Then,
\begin{align}\label{last}
    Pr\bigg(\exists \ell \in [n]: c(f_\ell)\leq \frac{1}{\ell\log^2 n}\bigg) \leq \sum_{\ell=1}^{n}\frac{1}{\ell \log^2n} =O\bigg(\frac{1}{\log n}\bigg).
\end{align}
\\ Observe that $T$ purchases $\ell$ $P_3$'s and then an edge from the set $f_1,...,f_\ell$ for some $\ell \in [n]$.
Hence, since $T$ is an arbitrary strategy, (\ref{last}) and Lemma \ref{lowhard} imply
\begin{align*} K_{C_4}^{POM}
&\geq \underset{}{\min}\Bigg\{ \underset{\ell \geq n^{0.5}}{\min}\bigg\{   \frac{\ell^{0.8}}{ 10^5n\log^{0.8}n}+\frac{1}{\ell\log^2n} \bigg\}, \underset{l<n^{0.5}}{\min}\bigg\{ \frac{1}{\ell\log^2n} \bigg\} \Bigg\} \\
&= \underset{\ell \geq n^{0.5}}{\min}\Bigg\{ \frac{\ell^{0.8}}{ 10^5 n\log^{0.8}n}  +\frac{1}{\ell\log^2n} \Bigg\} \geq 10^{-4}n^{-\frac{5}{9}}\log^{-\frac{4}{3}}n .
\end{align*}
\section{Final Remarks}
In this paper we analysed the minimum expected cost of purchasing a $C_4$ in the $POM$ and $ROM$ settings. In turns out that the two quantities differ by at most a multiplicative  factor of $\log^{\frac{4}{3}}n$. Furthermore, the lower bound that we proved of the cost of purchasing a $C_4$ in the $POM$ model holds w.h.p. In order to get an upper bound of the cost of purchasing a $C_4$ in the $ROM$ model that also holds w.h.p. we may alter our strategy for purchasing an edge from $END$ into the following one. Buy the first edge in $END$ that is examined and costs less than 
$\frac{\log n}{\vert END \vert}$. In this case we purchasing a cheap edge with probability $1-(1-{\log n}/{\vert END \vert })^{\vert END\vert}=1-o(1).$
\vspace{3mm}
\\The bounds that are proved in this paper can be extended to the case where the underlying graphs is $G_{{n,p}}$
with $p$ being constant
(instead of the complete graph on $n$ vertices, $K_n$). 
In this case it is straight forward to see that we can use the same methodology  in order to prove that bounds in the case where our underlying graph $G$ is $G_{n,p}$ are just a factor of $\frac{1}{p}$ larger than the corresponding ones in the case where  $G=K_n$. The $\frac{1}{p}$ factor arise from the fact that at cases of interest only a $p$ factor of  edges are present.
\vspace{3mm}
\\Finally, it would be of some interest to
\begin{itemize}
\item Close the $ O\big(\log^{\frac{4}{3}}n\big)$ gap of the cost of purchasing a $C_4$ in the $POM$ and $ROM$ settings.
    \item Analyse the cost of purchasing multiple $C_4$'s.
    \item Replace $C_4$ by other graphs.
\end{itemize}
\vspace{3mm}
{\large\bf{Acknowledgement:}} I thank Alan Frieze for his comments on the paper.


\begin{thebibliography}{9}


\bibitem{secr}
  M. Babaioff, N. Immorlica and R.Kleinberg,
 Matroids, secretary problems and on-line mechanisms, \emph{ Proccedings of the eighteenth annual ACM-SIAM Symposium on Discrete Algorithms }(2007) 434-443.



\bibitem{cal}
 A.\@ Cayley,
 Mathematical questions with their solutions,
\emph{ The Educational Times 23 }(1875) 18-19. See
\emph{The Collected Mathematical Papers of Arthus Cayley 10 } (1896) 587-588, Cambridge University Press.


\bibitem{Ch0}
Chernoff, Herman. A Measure of Asymptotic Efficiency for Tests of a Hypothesis Based on the sum of Observations,
\emph{Ann. Math. Statist. 23 }(1952) 493--507.


\bibitem{opt}
E.B.\@ Dynkin, The optimum choice of the instant for stopping Markov proccess,
\emph{Sov. Math. Dokl. 4} (1963).

\bibitem{fur}
T.S.\@ Ferguson, Who solved the secretary problem?
\emph{Statistical Science 4} (1989) 282-296.



\bibitem{pur}
  A.M.Frieze and W. Pegden,
Online purchasing under uncertainty, \emph{arXiv:1605.06072 [cs.DS]}


\bibitem{mos}
  L.\@ Moser,
  On a problem of Caley,
  \emph{ Scripta Mathematica 22 }(1956) 289-292.

\bibitem{Ch}
W.\@ Hoeffding, Probability inequalities for sums of bounded random variables, \emph{Journal of the American
Statistical Association 58} (1963) 13-30.


\end{thebibliography}
\end{document}